\RequirePackage{fix-cm}
\documentclass[smallextended]{svjour3}       
\smartqed  
\usepackage{graphicx}
\usepackage{amssymb,amsmath,latexsym,array}
\usepackage{multirow}
\newcolumntype{P}[1]{>{\centering\arraybackslash}m{#1}}
\begin{document}

\newcommand{\ob}{\mbox{$\overline{\omega}$}}
\newcommand{\om}{\mbox{$\omega$}}
\newcommand{\Tr}{\mbox{Tr}}
\newcommand{\C}{\mbox{$\cal C$}}
\newcommand{\Cperb}{\mbox{$\C^\bot$}}
\newcommand{\BF}{\mathbf}
\newcommand{\ben}{\begin{equation*}}
\newcommand{\een}{\end{equation*}}
\newcommand{\be}{\begin{equation}}
\newcommand{\ee}{\end{equation}}
\newcommand{\Go}{\Gamma_{H_1}}
\newcommand{\Gt}{\Gamma_{H_2}}
\newcommand{\Gtt}{\Gamma_{H_3}}

\title{Secret sharing schemes based on additive codes over $GF(4)$ }


\author{Jon-Lark Kim         \and
        Nari Lee 
}


\institute{Jon-Lark Kim \at
              Department of Mathematics \\ Sogang University  \\ Seoul, 121-742, South Korea \\
 \email{jlkim@sogang.ac.kr}
           \and
           Nari Lee \at
             Department of Mathematics \\
Sogang University \\
Seoul 121-742, South Korea \\
 \email{narilee3@gmail.com}
}

\date{Received: date / Accepted: date}

\maketitle

\begin{abstract}
A secret sharing scheme (SSS) was introduced by Shamir in 1979 using polynomial interpolation. Later it turned out that it is equivalent to an SSS based on a Reed-Solomon code. SSSs based on linear codes have been studied by many researchers. However there is little research on SSSs based on additive codes. In this paper, we study SSSs based on additive codes over $GF(4)$ and show that they  require at least two steps of calculations to reveal the secret. We also define minimal access structures of SSSs from additive codes over $GF(4)$ and  describe SSSs using some interesting additive codes over $GF(4)$ which contain generalized 2-designs.
\end{abstract}

\section{Introduction}

A $secret~ sharing~ scheme$ (SSS) is a method of distributing a secret to a finite set of participants such that only predefined subsets of the participants can recover the secret.  All the participants receive a piece of the secret, known as a $share$, in such a way that only qualified subsets of the participants can have access to the secret by pooling the shares of their members.

The first construction of an SSS was done by Shamir \cite{Sha} and Blakley \cite{Bla} independently in 1979. Shamir used polynomial interpolation for constructing an SSS, while Blakley used hyperplane geometry. Later Shamir's SSS turned out to be equivalent to an SSS based on a Reed-Solomon code \cite{McE}. Some of SSSs were applied to various fields such as cloud computing, controlling nuclear weapons in military, recovering information from multiple servers, and controlling access in banking system.

Since an SSS plays an important role in protecting secret information, it has been studied by several authors (see \cite{Ding}, \cite{Dou}, \cite{Li}, and \cite{Massey}). In particular, Massey \cite{Massey} used  linear codes for secret sharing and pointed out the relationship between the access structure and the minimal codewords of the dual code of the underlying code in 1993.  However there has been less attention to SSSs based on additive codes. Since additive codes include linear codes, we raise an intriguing question whether SSSs based on additive codes have more advantage than SSSs based on linear codes.

Ding et al. \cite{Ding} remark the following.  Normally the weight distribution of a code is very hard to determine and that of only a few classes of codes is known. As an SSS can be constructed from any error-correcting linear code, what matters is how we are going to determine the access structure. The access structure of SSSs based on error-correcting codes depends on the weight distribution of their dual codes. To determine the access structure for the SSSs we need more information than the weight distribution. This makes it difficult to determine the access structure of SSSs based on codes, as determining the weight distribution of codes is difficult.

In this paper we introduce SSSs based on linear codes with an example in Section 2. In Section 3.2, we  define SSSs based on additive codes and show that they require two steps of calculations, while the SSSs based on linear codes require only one step of calculation.  In Section 3.3,  we define minimal access structures of SSSs based on additive codes over GF(4), which is developed from Proposition 2 in \cite{Ding2}. We also describe SSSs based on a hexacode, a dodecacode over $GF(4)$ which is described and analyzed in \cite{Lark}, and $S_{18}$. We determine the access structure of the SSSs and prove their properties.  The access structure for these SSSs is more abundant than that of the SSSs based on linear codes. We are able to determine the access structure of these SSSs because the structure of the underlying  additive codes is thoroughly understood.

\section{Some preliminaries}
\label{sec:intro}

\ Let $GF(q)$ be a finite field with $q=p^r$ elements, where $p$ is a prime and $r$ is a positive integer.  The $Hamming\ distance$ between two vectors $\bold{x,y}\in GF(q)^n$ is defined to be the number of coordinates in which $\bold{x}$ and $\bold{y}$ differ. Note that the $minimum\ distance$ $d$ of a code $C$ is the smallest nonzero distance between two distinct codewords and is important in determining the error-correcting capability of $C$; the higher the minimum distance, the more errors the code can correct. In general, it can detect up to $d-1$ errors and correct up to  $\lfloor \frac{d-1}{2} \rfloor$ errors.

The $Hamming\ weight$ of a vector $\bold{c}$ in $GF(q)^n$ denoted by $wt(\bold{c})$  is the total number of nonzero coordinates. Let $A_i$, also denoted by $A_i(C)$, be the number of codewords of weight $i$ in $C$. The list of  $A_i$ for $0\leq i \leq n $ is called the $weight\ distribution$ of $C$. An $[n,k,d]$ code $C$ is a  linear subspace of $GF(q)^n$ with dimension $k$ and minimum nonzero Hamming weight $d$. A generator matrix $G$ for an $[n,k,d]$ code $C$ is any $k\times n$ matrix $G$ whose rows form a basis for $C$.

\ We refer to  \cite{Yuan}. Let $G=(\bold{g}_0, \bold{g}_1,\cdots, \bold{g}_{n-1})$ be a generator matrix of an $[n,k,d]$ code. We assume that none of $\bold{g}_i$'s is the zero vector.  Let $\bold{g}_i~(1\leq i \leq n-1)$  be a column vector. In an SSS constructed from an $[n,k,d]$ linear code $C$, the secret is an element of $GF(q)$ and there are $n-1$ participants $P_1, P_2, \cdots, P_{n-1}$ and a dealer $P_0$. To compute the shares with respect to the secret $s$, the dealer randomly takes an element $\bold{u}$=$(u_0,u_1,\cdots, u_{k-1})\in GF(q)^k$ such that $s=\bold{ug}_0$. The dealer  treats $\bold{u}$ as an information vector and computes the corresponding codeword
\begin{center}
$\bold{t}=(t_0,t_1,\cdots,t_{n-1})$=$\bold{u}G$
\end{center}
and the dealer gives the share $t_i$ to participant $P_i$ as share for each $i\geq 1$.\\

Since $t_0=\bold{ug}_0=s$, a set of shares $\{t_{i_1}, t_{i_2},\cdots ,t_{i_m}\}$ determines the secret $s$ if and only if the column $\bold{g}_0$ of the generating matrix $G$ is a linear combination of the columns $\{g_{i_1},g_{i_2},\cdots , g_{i_m}\}$ of $G$.

\begin{lemma} \em{ (\cite{Yuan})} \label{lemma1} \it{Let $C$ be an $[n,k,d]$ linear code over the finite field $GF(q)$ and let $C^{\perp}$ be its dual code. In the secret sharing scheme based on $C$, a subset of shares $\{t_{i_1}, t_{i_2},\cdots ,t_{i_m}\}$, $1\leq i_1 \leq \cdots \leq i_m\leq n-1$, determines the secret if and only if there is a codeword
\begin{equation}\label{L1}
(1,0, \cdots ,0,c_{i_1},0, \cdots ,0,c_{i_m},0, \cdots ,0 )
\end{equation}
in $C^{\perp}$ with $c_{i_j} \neq 0$ for at least one $j$.}
\end{lemma}

We explain how the secret is recovered using (\ref{L1}). If there is a codeword of (\ref{L1}) in $C^{\perp}$, then the vector $\mathbf{g}_0$ is a linear combination of $\BF{g}_{i_1},\ldots, \BF{g}_{i_m}$ , i.e.,
\ben
\mathbf{g}_0=\sum_{j=1}^{m}x_j \BF{g}_{i_j},
\een
where $x_j\in GF(q)$ for $1\leq j\leq m$.

Then the secret $s$ is recovered by computing
\ben
s=\sum_{j=1}^{m}x_j t_{i_j}.
\een

Here we think about a case of some malicious behaviors lying among participants, called $cheaters$. They modify their shares in order to cheat.  In this case, we make use of  detection of errors up to $d-1$ and correction of errors up to $\lfloor \frac{d-1}{2}\rfloor$.   The errors being considered as modified shares, SSSs based on error-correcting codes are able to detect up to $d-1$ cheaters and correct up to    $\lfloor \frac{d-1}{2} \rfloor$  cheaters.

If a group of participants can recover the secret by pooling their shares, then any group of participants containing this
group can also recover the secret.

\begin{definition}\label{def1}

\rm
An $access$  $group$ is a subset of a set of participants that can recover the secret from its shares. A collection $\Gamma$ of access groups  of participants is called an $access$  $structure$ of the scheme.  An element $A\in \Gamma$ is called a $minimal\ access\  group$  if no element of  $\Gamma$ is a proper
subset of $A$. Hence a set is a minimal access group if it can
recover the secret but no proper subset can recover the secret.
We let  $\bar{\Gamma}= \{A | A~\textrm{ is a minimal access group }\}$. We call
$\bar{\Gamma}$  the $minimal\ access\ structure$.

\end{definition}
In general, determining the
minimal access structure is a difficult problem \cite{Ding}.

\begin{definition}\label{def2}\rm The $support$ of a vector ${\bf c}=(c_0,\cdots,c_{n-1})\in GF(q)^n$ is defined by
\begin{center}
$supp({\bf c})=\{0\leq i \leq n-1\ |\ c_i\neq0\}$.
\end{center}
Let ${\bf c_1}$ and ${\bf c_2}$ be two codewords of a code $C$. We say that ${\bf c_1}$ $ covers $ ${\bf c_2}$ if $supp({\bf c_2}) \subseteq supp({\bf c_1})$.

If a nonzero codeword ${\bf c}$ covers only its scalar multiples, but no other codewords, then it is called a $minimal$ $codeword$.
\end{definition}

\begin{theorem} \textrm{\em{\cite{Ding2}}}\label{mini}
 Let $C$ be an $[n,k;q]$ code, and let $G=(\bold{g}_0, \bold{g}_1,\cdots, \bold{g}_{n-1})$ be its generator matrix, where all $\bold{g}_i$'s are nonzero.  If each nonzero codeword of $C$ is minimal, then in the secret sharing scheme based on $C^{\perp}$, there are altogether $q^{k-1}$ minimal access groups. In addition, we have the following:
\begin{itemize}
\item[$\bullet$] If $\bold{g}_i$ is a scalar multiple of  $\bold{g}_0$, $1\le i \le n-1$, then participant $P_i$ must be in every minimal access set. Such a participant is called a  $dictatorial$ $participant$.

\item[$\bullet$] If $\bold{g}_i$ is not a scalar multiple of  $\bold{g}_0$, $1\le i \le n-1$, then participant $P_i$ must be in $(q-1)q^{k-2}$ out of $q^{k-1}$ minimal access groups.

\end{itemize}
\end{theorem}

\begin{definition}\label{def3}
\rm A $t$-$(v,k,\lambda)$ $design$ or briefly a $t$-$design$, is a pair $(\mathcal{P,B})$ where $\mathcal{P}$ is a set of $v$ elements, called $points$, and $\mathcal{B}$ is a collection of distinct subsets of $\mathcal{P}$ of size $k$, called $blocks$, such that every subset of points of size $t$ is contained in precisely $\lambda$ blocks.
\end{definition}
For linear codes with a special weight distribution, a powerful result of the Assmus-Mattson theorem guarantees that a set of codewords with a fixed weight holds a $t$-design \cite{Ass}. The Assmus-Mattson Theorem has been the main tool in discovering designs in codes.

\begin{theorem}\textrm{(Assmus-Mattson \em{\cite{Huf}})}\label{theorem1}
Let $C$ be an $[n,k,d]$ code  over $GF(q)$. Suppose $C^{\perp}$ has minimum weight $d^{\perp}$. Let $w$ be the largest integer with $w\le n$ satisfying
\begin{equation*}
w- \lfloor \frac{w+q-2}{q-1} \rfloor < d.
\end{equation*}
(So $w=n$ when $q=2$.) Define $w^{\perp}$ analogously using $d^{\perp}$.  Suppose that $A_i=A_i(C)$ and $A^{\perp}_i=A_i(C^{\perp})$, for $0\le i \le n$, are the weight distributions of $C$ and $C^{\perp}$, respectively. Fix a positive integer $t$ with $t<d$, and let $s$ be the number of $i$ with $A_i^{\perp} \ne 0$  for $0<i\le n-t$. Suppose $s\le d-t$. Then:
\begin{itemize}
\item[(i)] the vectors of weight $i$ in $C$ hold a $t$-design provided $A_i\ne 0$ and $d\le i \le w$, and
\item[(ii)] the vectors of weight $i$ in $C^{\perp}$ hold a $t$-design provided $A^{\perp}_i \ne 0$ and $d^{\perp} \le i \le$ min$\{n-t, w^{\perp}\}$.
\end{itemize}
\end{theorem}

A $t$-$(v,k,\lambda)$ design is also an $i$-$(v,k,\lambda_i)$ design for $0\leq i\leq t$. We can get $\lambda_i$ by the formula \cite{Huf}:
\begin{equation}\label{eq2}
\lambda_i=\lambda \binom{v-i}{t-i}/\binom{k-i}{t-i}.
\end{equation}

Let $W_{\C}=\sum A_i y^i$.  Let $D_i$ denote the 1-design formed from the vectors of weight $i$, and $\lambda_s(D_i)$  denote $\lambda_s$ for that particular design. Here $\lambda_s$ denotes the number of blocks that are incident with a given $s$-tuple of points  for $s\le t$.
\begin{corollary}\em{(\cite{Dou})}\label{corollary1}
\it{The access groups in the secret sharing scheme based on a binary self-dual code $\C$  have  the following size distribution generating function when groups of each size form a 1-design :}
\begin{equation}
 \sum_i \lambda_1(D_i) y^{i-1}.
\end{equation}
\end{corollary}

Now let us consider the accessibility of an access structure. Let  $P = \{P_1,\ldots, P_m\}$ be a set of $m$ participants and let $\mathcal{A}_P$ be the set of all access structures on  $P$.
\begin{definition}(\cite{Car})\label{def4}
\rm{The $accessibility~ index$ on $P$ is the map $\delta_{P} : \mathcal{A}_P  \longrightarrow \mathbb{R}$   given by}
\begin{equation*}
\delta_{P}(\Gamma)=\frac{|\Gamma|}{2^m}~\textrm{for}~\Gamma \in \mathcal{A}_P
\end{equation*}
\rm{where $m=|P|$. The number $\delta_{P}(\Gamma)$  will be called the $accessibility~ degree$ of structure $\Gamma$.}
\end{definition}

$\delta_{P}(\Gamma)$ may be interpreted as the probability of a random coalition in $P$  to be
authorized when each participant has a probability $1/2$ to belong to it. As it is
obvious, $\delta_{P}(\Gamma)=0$ iff $\Gamma=\emptyset$.  Otherwise, $0 < \delta_{P}(\Gamma) < 1$, and $|\Gamma| < |\Gamma '|$ implies $\delta_{P}(\Gamma)< \delta_{P}(\Gamma ')$.

Since supports of each weight in a code holding a 1-design determine the size of the access structure $\Gamma$, the accessibility degree of $\Gamma$  for the SSS based on the code can be defined as follows   :
\begin{equation*}
\delta_{P}(\Gamma)=\frac{1}{2^m}\sum_{i}\lambda_1(D_i).
\end{equation*}

\begin{remark}
The Gleason-Pierce-Ward Theorem \cite{Huf} provides the main motivation for studying self-dual
codes over $GF(2)$, $GF(3)$, and $GF(4)$ since these codes have the property that they are divisible. When a code C is divisible by $c >1$, it  implies that all  codewords have weights divisible
by an integer $c$, which is called a divisor of C.
\end{remark}

\begin{example}
Here we introduce an SSS from
$[24, 12, 8]$ Golay code as given in \cite{Dou}. The weight enumerator of the length 24 Golay code is:
\begin{equation}
1+759y^8+2576y^{12}+759y^{16}+y^{24}.
\end{equation}
Note  that the supports of any nonzero
weight of the $[24,12,8]$ Golay code form a $5$-design. It is  easy to calculate for $\lambda_1(D_8)=253$, $\lambda_1(D_{12})=1288$, and $\lambda_1(D_{16})=506$.
These groups together with the entire group, comprise
the 2048 elements of the access structure. Each of the 253
groups of size 8 must be in the minimal access structure.
Additionally, each of the 1288 groups of size 12 must be
in the minimal access structure because if the support of a weight 8
vector were a subset of the support of a weight 12 vector
then the sum of these vectors would have weight 4, which is
a contradiction. The group of size 24 is not in the
minimal access structure. We note that no weight 16 vector
can have a support containing the support of weight 12 vector
since it would produce a weight 4 vector in the code, which is a
contradiction. There are 253 weight
16 vectors whose support cannot be in the minimal access
structure and 253 that are in the minimal access structure.
This gives the following.

\begin{theorem}\label{theorem2}
\it{\em{(\cite{Dou})} In the secret sharing scheme produced from
the extended Golay code we have the following  :}
\begin{itemize}
\item[$\bullet$] \it{The access structure consists of 253 groups of size 7,
1288 groups of size 11, 506 groups of size 15 and 1
group of size 23.}
\item[$\bullet$] \it{The minimal access structure consists of the 253 groups
of size 7, the 1288 groups of size 11, and 253 groups of
size 15.}
\item[$\bullet$] \it{No group of size less than 7 can determine the secret.}
\end{itemize}
\end{theorem}
\end{example}

\section{SSSs based on additive codes over $GF(4)$}

\subsection{Introduction to additive codes over $GF(4)$}

An {\em additive code \C\ over $GF(4)$ of length $n$} is
an additive subgroup of $GF(4)^n$ (see \cite{Lark} for details). Since \C\ is a vector
space over $GF(2)$, it has a basis consisting of $k~( 0 \le k \le 2n)$
vectors whose entries are in $GF(4)$. We call \C\ an $(n, 2^k)$ code.
A {\em generator matrix} of \C\ is a $k \times n$ matrix with
entries in $GF(4)$ whose rows are a basis of \C.\
The {\em weight}  of {\bf c}, denoted as  wt({\bf c}), in \C\ is the number of
nonzero components of {\bf c}. The minimum weight $d$ of \C\
is the smallest weight of any nonzero codeword in \C.\ If \C\
is an $(n,2^k)$ additive code of minimum weight $d$, \C\ is
called an $(n,2^k,d)$ code.
In order to define an inner product on additive codes
we define the {\em trace} map,
i.e., for $x$ in $GF(4)$, $\Tr (x)=x+x^2 \in GF(2)$.
We now define the {\em trace inner product} of two vectors
${\bf x}=(x_1x_2\cdots x_n)$ and ${\bf y}=(y_1y_2\cdots y_n)$ in $GF(4)^n$
to be
\[{\bf x}\star{\bf y}=\sum_{i=1}^n\Tr(x_i\overline{y_i}) \in GF(2),\]
where $\overline{y_i}$ denotes the conjugate of $y_i$.
Note that $\Tr(x_i \overline{y_i})=1$ if and only if
$x_i$ and $y_i$ are nonzero distinct elements in $GF(4)$.

If \C\ is an additive code, its {\em dual}, denoted by \Cperb, is the
additive code $\{ {\bf x} \in GF(4)^n \mid {\bf x} \star {\bf c}=0
\mbox{ for all }{\bf c} \in \C\}$. If \C\ is an $(n,2^k)$ code,
then \Cperb\ is an $(n,2^{2n-k})$ code. As usual, \C\ is called
{\em self-dual} if $\C = \Cperb$. We note that if \C\ is self-dual,
\C\ is an $(n,2^n)$ code.

\subsection{SSSs based on additive codes over $GF(4)$}
\ Let $G=(\bold{g}_0, \bold{g}_1,\cdots, \bold{g}_{n-1})$ be a generator matrix of an $(n,2^k)$ code over $GF(4)$, where $\bold{g}_i$ denotes the generic column of $G$. We assume that none of $\bold{g}_i$'s is the zero vector.  In an SSS constructed from an  $(n,2^k)$ code $C$, the secret is an element of $GF(4)$, and $n-1$ participants $P_1, P_2, \cdots, P_{n-1}$ and a dealer $P_0$ are involved. To compute the shares with respect to the secret $s$, the dealer randomly takes an element $\bold{u}$=$(u_0,u_1,\cdots, u_{k-1})\in GF(2)^k$ such that $s=\bold{ug}_0$. There are altogether $2^{k-2}$ vectors $\bold{u} \in GF(2)^k$ if $s$ is in $GF(4)$, or $2^{k-1}$ if $s$ an element in $\{0,1\}$, $\{0, \om\}$, or $\{0, \ob\}$. The dealer then treats $\bold{u}$ as an information vector and computes the corresponding codeword
\begin{center}
$\bold{t}=(t_0,t_1,\cdots,t_{n-1})$=$\bold{u}G$
\end{center}
and the dealer gives the share $t_i$ to participant $P_i$ as share for each $i\geq 1$.\\

\begin{lemma}\label{lem2}   Let $\C$ be an $(n,2^k)$ code over $GF(4)$ and  $\C^{\perp}$  its dual code defined by the trace inner product. Let
\begin{equation}\label{eqlem}
\begin{split}
&H_1=\left\{x|x=(1, \cdots ,0,x_{i_1},0, \cdots ,0,x_{i_m},0, \cdots ,0 )\in\Cperb \right. \\
 &\hspace{4cm}\left. \textrm{ with $x_0=1$, $x_{i_j} \neq 0$ for at least one $j$} \right\},\\
 &H_2=\left\{y|y=(\om, \cdots ,0,y_{i_1},0, \cdots ,0,y_{i_l},0, \cdots ,0 )\in\Cperb \right. \\
&\hspace{4cm}\left. \textrm{ with $y_0=\om$, $y_{i_j} \neq 0$ for at least one $j$} \right\}, \\
 &H_3=\left\{z|z=(\ob, \cdots ,0,z_{i_1},0, \cdots ,0,z_{i_r},0, \cdots ,0 )\in\Cperb\right. \\
&\hspace{4cm}\left. \textrm{ with $z_0=\ob$, $z_{i_j} \neq 0$ for at least one $j$} \right\}.
\end{split}
\end{equation}

In the secret sharing scheme based on $C$,  two subsets of shares $\{t_{i_1}, t_{i_2},\cdots ,t_{i_m}\}$ and  $\{t_{i_1}, t_{i_2},\cdots ,t_{i_l}\}$,  for  $1\leq i_1 < \cdots < i_m\leq n-1$ and $1\leq i_1 < \cdots < i_l\leq n-1$, determine the secret if and only if there are at least two codewords from distinct sets among $H_i$\rq{}s, $ 1\leq i \leq 3$.
\end{lemma}

\begin{proof}
$(\Leftarrow)$  Suppose there are at least two codewords in $\C^{\perp}$ as in (\ref{eqlem}). Then we get two of the three equations from trace inner product as follows:

\begin{equation}\label{pf1}
\begin{split}
&(t_{0}+t_{0}^{2})+(t_{i_1}\bar{x_1}+(t_{i_1}\bar{x}_{1})^{2})+\ldots+(t_{i_m}\bar{x_m}+(t_{i_m}\bar{x_m})^{2})=0,\\
&(t_{0}\ob+(t_{0}\ob)^{2})+(t_{i_1}\bar{y_1}+(t_{i_1}\bar{y_1})^{2})+\ldots+(t_{i_l}\bar{y_l}+(t_{i_l}\bar{y_l})^{2})=0,\\
&(t_{0}\om+(t_{0}\om)^{2})+(t_{i_1}\bar{z_1}+(t_{i_1}\bar{z_1})^{2})+\ldots+(t_{i_r}\bar{z_r}+(t_{i_r}\bar{z_r})^{2})=0.\\
\end{split}
\end{equation}

Since $s=ug_{0}=t_{0}$, the equation can be rewritten as
\begin{equation*}
\begin{split}
&s+s^{2}= \sum_{j=1}^{m}(t_{i_j}\bar{x_j}+(t_{i_j}\bar{x}_{j})^{2})\in GF(2),\\
&s\ob+(s\ob)^{2}= \sum_{j=1}^{l}(t_{i_j}\bar{x_j}+(t_{i_j}\bar{x}_{j})^{2})\in GF(2),\\
&s\om+(s\om)^{2}= \sum_{j=1}^{l}(t_{i_j}\bar{x_j}+(t_{i_j}\bar{x}_{j})^{2})\in GF(2).\\
\end{split}
\end{equation*}

Let $\alpha_1=s+s^{2}$, $\alpha_2=s\ob+(s\ob)^{2}$,  and $\alpha_3=s\om+(s\om)^{2}$.
Now the secret $s$ can be recovered using two values of $\alpha_i$'s based on Table 1.  For example, if $\alpha_1=0$ and $\alpha_2=1$, then the secret $s$ is uniquely determined as 1.


As we can see in Table 1, we do not need all the three values of $\alpha_i$'s since two values of $\alpha_i$'s are sufficient to determine the secret $s$.  Now we can say we recover the secret $s$  with two values of $\alpha_i$'s, where $i=1,2$, or $3$.

\begin{table}[h]\label{t1}
\caption{Recovering the secret $s$ from $\alpha_i$'s}
\centering
\begin{tabular}[c]{|P{2.5cm}P{2.5cm}P{2.5cm}|P{1.5cm}|}
\hline
$\alpha_1=s+s^{2}$&$\alpha_2=s\ob+(s\ob)^{2}$&$\alpha_3=s\om+(s\om)^{2}$&$s$\\ \hline
0&0&0&0\\
0&1&1&1\\
1&0&1&\om \\
1&1&0&\ob \\\hline
\end{tabular}
\end{table}

\newpage
\noindent
$(\Rightarrow)$  Suppose there are two subsets of shares $\{t_{i_1}, t_{i_2},\cdots ,t_{i_m}\}$ and  $\{t_{i_1}, t_{i_2},\cdots ,t_{i_l}\}$,  for  $1\leq i_1 < \cdots < i_m\leq n-1$ and $1\leq i_1 < \cdots < i_l\leq n-1$, that determine the secret $s$.

First, we note the following: 

\begin{equation}
\begin{split}
(\bold{ug}_i)^2&=((u_0,u_1,\cdots, u_{k-1})(g_{0i},g_{1i}, \ldots, g_{(k-1),i})^{T})^2\\
&=(u_0g_{0i}+u_1g_{1i}+\ldots +u_{k-1}g_{(k-1),i})^2\\
&=u_0^2g_{0i}^2+u_1^2g_{1i}^2+\ldots +u_{k-1}^2g_{(k-1),i}^2\\
&=(u_0^2+u_1^2+\ldots +u_{k-1}^2)(g_{0i}^2+g_{1i}^2+\ldots +g_{(k-1),i}^2)\\
&=(u_0^2,u_1^2,\ldots ,u_{k-1}^2)(g_{0i}^2,g_{1i}^2, \ldots, g_{(k-1),i}^2)^{T}\\
&=(u_0,u_1,\cdots, u_{k-1})(g_{0i}^2,g_{1i}^2, \ldots, g_{(k-1),i}^2)^{T}\\
\end{split}
\end{equation}

Here  $\bold{u}=(u_0,u_1,\cdots, u_{k-1})=(u_0^2,u_1^2,\ldots ,u_{k-1}^2)$ since $\bold{u}\in GF(2)^k$. Letting $\bold{g}_i^2=(g_{0i}^2,g_{1i}^2, \ldots, g_{(k-1),i}^2)^{T}$ for convenience, we have    $(\bold{ug}_i)^2=\bold{ug}_i^2$.

Now we can rewrite $\alpha_i$'s in the following way :
\begin{equation*}
\begin{split}
&\alpha_1=\sum_{j=1}^{m} \left(t_{i_j}\bar{x_j}+( t_{i_j}\bar{x_j})^2\right)=\bold{u}\sum_{j=1}^{m}(x_j\bold{g}_{i_j}+\bar{x_j}\bold{g}_{i_j}^2),\\
&\alpha_2=\sum_{j=1}^{l} \left(t_{i_j}\bar{y_j}+( t_{i_j}\bar{y_j})^2\right)=\bold{u}\sum_{j=1}^{l}(x_j\bold{g}_{i_j}+\bar{x_j}\bold{g}_{i_j}^2),\\
&\alpha_3=\sum_{j=1}^{r} \left(t_{i_j}\bar{z_j}+( t_{i_j}\bar{z_j})^2\right)=\bold{u}\sum_{j=1}^{r}(x_j\bold{g}_{i_j}+\bar{x_j}\bold{g}_{i_j}^2).\\
\end{split}
\end{equation*}

We can determine two of the values of $\alpha_i$'s, $1\leq i \leq 3$ by the two sets of shares, $\{t_{i_1}, t_{i_2},\cdots ,t_{i_m}\}$ and  $\{t_{i_1}, t_{i_2},\cdots ,t_{i_l}\}$,  for  $1\leq i_1 < \cdots < i_m\leq n-1$,  $1\leq i_1 < \cdots < i_l\leq n-1$,  if and only if
\begin{equation}
\begin{split}
\bold{g}_0+\bold{g}_0^2=\sum_{j=1}^{m}(x_j\bold{g}_{i_j}+\bar{x_j}\bold{g}_{i_j}^2), ~~&\bar{\omega}\bold{g}_0+\omega\bold{g}_0^2=\sum_{j=1}^{l}(x_j\bold{g}_{i_j}+\bar{x_j}\bold{g}_{i_j}^2),\\
& \omega\bold{g}_0+\bar{\omega}\bold{g}_0^2=\sum_{j=1}^{r}(x_j\bold{g}_{i_j}+\bar{x_j}\bold{g}_{i_j}^2).
\end{split}
\end{equation}

We can find $x_j$'s by solving the linear equations and get two  values of $\alpha_i$'s. Using these two values of $\alpha_i$'s we can recover the secret $s$ by Table 1.
Hence  there exist at least two
codewords from distinct sets among $H_i$'s, $1\leq i \leq 3$ if we recover the secret $s$ using two subsets of shares, $\{t_{i_1}, t_{i_2},\cdots ,t_{i_m}\}$ and  $\{t_{i_1}, t_{i_2},\cdots ,t_{i_l}\}$.  \qed

\end{proof}

Now we need to define an access group and an access structure for SSS based on additive codes over $GF(4)$.



\noindent
Let
\begin{equation*}
\begin{split}
&\Go=\{\textrm{the set of supports for $x\in H_1$ excluding 1 from each support}\} , \\
&\Gt=\{\textrm{the set of supports for $y\in H_2$ excluding 1 from each support}\}, \\
 &\Gtt=\{\textrm{the set of supports for $z\in H_3$ excluding 1 from each support}\}.\\
\end{split}
\end{equation*}
 The access structure for a linear code based SSS is a set of supports of vectors in $\Cperb$ with $c_0=1$, which is same to $\Go$. The access structures from additive codes over $GF(4)$ are different from those from linear codes. To recover the secret $s$, we need at least two sets among $\Go,~\Gt$, or $\Gtt$  for an access structure. We obtain the values of $\alpha_i$ and $\alpha_j$ from two elements of $\Gamma_{H_i}$ and $\Gamma_{H_j}$, $i \ne j$, respectively. With the two values, we can recover the secret $s$ using the table above.

Since this process  requires at least two steps of calculations to reveal the secret, we call this process as a $2$-$step$ SSS. On the other hand, the previous SSS can be regarded as a $1$-$step$ SSS.

We need the Assmus-Mattson Theorem for additive codes over $GF(4)$ which gives designs with possibly repeated blocks to define our process.

\begin{theorem}\em{(\cite{Lark})}\label{Larkt}
\it{Let $\C$ be an additive $(n, 2^k)$ code over $GF(4)$ with minimum weight $d$. Let $\Cperb$ be its dual $(n, 2^{2n-k})$ code with minimum weight $d'$. Let $0<t<d$. Let $s$ be the number of weights $B_i \neq 0$ in $\Cperb$ where $0<i\leq n-t$. Suppose that $s\leq d-t$. Then the following hold.}
\begin{enumerate}
\item[(i)] \it{For each weight $u~(d\leq u\leq n)$, the set of supports of codewords of weight $u$ in $\C$ holds a $t$-design with possibly repeated blocks.}
\item[(ii)]\it{The set of supports of vectors of weight $w$ in $\Cperb$ where $B_w \neq 0$ and $d' \leq w\leq n-t$ hold a $t$-design with possibly repeated blocks.}
\item[(iii)] \it{The supports of minimum weight vectors are either simple blocks or have repetition number 3.}
\end {enumerate}
\end{theorem}

\begin{corollary}\em{(\cite{Lark})}\label{Larkc}
\it{Let $n_i:=6m+2(i-1)$ with $m\geq 1$ any integer and $i=1,2,$ or $3$.
Let $\C$ be an extremal additive even self-dual $(n_i, 2^{n_i})$ code over $GF(4)$ with minimum
weight $d=2m+2\geq 6$. Then the vectors of each weight $w$ in $\C$ where $A_w\neq 0$ and
$d\leq w\leq n_i$ hold a $(7-2i)$-design with possibly repeated blocks.}
\end{corollary}

\begin{lemma}\label{lem1}
Let $C$ be an additive even $(n,2^k)$   self-dual code over $GF(4)$. Then the supports of codewords for all non-trivial
weights hold a 1-design with possible repeated blocks if $d \geq \frac{n+2}{3}$.
\end{lemma}
\begin{proof}
An additive $(n,2^k)$  self-dual code over $GF(4)$ has $\frac{n}{2}-1$ possible non-trivial weights. Then $\frac{d}{2}-1$ of these possible weights have no vectors since $d$ is the minimum weight. Since $d\geq \frac{n+2}{3}$,  from this we get the following inequality which satisfies Assmus-Mattson Theorem:
{\small \begin{equation*}
d-1\geq (\frac{n}{2}-1)-(\frac{d}{2}-1).
\end{equation*} }
Thus the supports of codewords for all non-trivial weights hold a 1-design with possibly repeated blocks.
\end{proof}

\subsection{SSSs based on extremal additive even self-dual  codes over $GF(4)$}

Up to now we have introduced a different method of defining access structures based on additive codes. We have shown that the access structures are nicely constructed in this way. However there might be some repeated blocks in these access structures. We, hence, employ the notion of a generalized $t$-design  from \cite{Del} to resolve this issue.

The generalized $t$-design is to count the number of groups in an access structure for an additive code based SSS. We will, first of all, redefine $covering$ an element for the generalized $t$-design \cite{Del}.

\begin{definition}
Let $G=GF(4)$, the set of $n$-tuples of $GF(4)$.
An element ${\bf a}$ of $G$ is said to be $componentwisely$ $covered$  (abbr. $c$-$covered$) by an element ${\bf b}$ of $G$ if each nonzero
component $a_i$ of ${\bf a}$ is equal to the corresponding component $b_i$ of ${\bf b}$; we denote this by ${\bf a}\leq {\bf b}$. For example,  ${\bf a}=(1,1,\om,0)$ is c-covered by ${\bf b}=(1,1,\om,\ob)$.
\end{definition}

\begin{definition}
A subset $S$ of $G$ is called a $generalized\ t$-$design\ of\ type\ q-1$, with parameters $t$-$(n,k,\mu_t),~0 \leq t\leq k\leq n$, $\mu_t\geq 1$, if the following two conditions are satisfied:
\begin{enumerate}
\item[(i)] all elements of $S$ have the same weight $k$,
\item[(ii)] each element of weight $t$ in $G$ is c-covered by a constant number $\mu_t$ of elements
of $S$. If a subset $S$ of $G$ holds a generalized $t$-design of type $q-1$, then it holds
a generalized  $(t-1)$-design of type $q-1$.
\end{enumerate}
\end{definition}

In the binary case ($q=2$), this is the same as a classical $t$-design without repeated
blocks.

For a given code $\C$ of length $n$ and for an element $e$ of $G=GF(4)^n$, we denote by $\mu(p,e)$
the number of codewords of weight $p$ that c-cover $e$ and $\mu_i(p,e)$ the number of codewords of weight $p$ in $\Gamma_{H_i}$ that c-cover $e$. Trivially, if $p< wt(e)$, then $\mu(p,e)=0$ and $\mu_i(p,e)=0$ .

Delsarte's theorem for any finite alphabet is given as follows.
\begin{theorem}\em{(\cite{Del})}\label{Del}
\it{Let $\C$ be a $q-$ary code of dual distance $d'$. Let $t$ be an
integer, $1\leq t\leq d'$, such that the number of weights of $\C$ that are at least equal to $t$ is at
most equal to $d'-t$. Then each set of codewords of a given weight $\geq t$ is a generalized
$t-$design of type $q-1$.}
\end{theorem}

Now we obtain generalized $t-$designs from additive codes over $GF(4)$.
\begin{corollary}\em{(\cite{Hoh})}\label{Hoh}
\it{Let $\C$ be an extremal even additive self-dual code over $GF(4)$ of
length $n=6m$ (respectively, $n=6m+2$). Then the set of codewords of weight $w$ in $\C$ with
$A_w\neq 0$ forms a generalized $2-$design (respectively, $1-$design) of type $3$.}
\end{corollary}

Since two elements from two different sets among $\Gamma_{H_i}$'s, $i\in\{1,2,3\}$,  are sufficient to recover a secret, we are going to consider all the  combinations of only  two  distinct $\Gamma_{H_i}$'s when defining access structures.  That is, all the pairs $(s_i, s_j)$ for $s_i \in \Gamma_{H_i}$ and $s_j \in \Gamma_{H_j}$ with $i \ne j$, comprise all the elements of the access structure. An element $(s_i, s_j)\in \Gamma$ is called a $minimal~ access ~group$ if neither $s_i$ nor $s_j$ c-covers any other elements in $\Gamma_{H_i}$ and $\Gamma_{H_j}$, respectively. We let $\bar{\Gamma}=\{(s_i, s_j)| (s_i, s_j) \textrm{ is a minimal access group}\} $ and call it by the $minimal ~ access ~ structure$.

\begin{theorem}\label{thm1}
The access structure of this secret sharing scheme is given by
\begin{equation}
\Gamma=\{(x,y)| x\in \Gamma_{H_i} ~\textrm{and}~ y\in \Gamma_{H_j}, ~\textrm{where}~ i\neq j~\textrm{and}~ i,j\in\{1,2,3\}\}.
\end{equation}
The number of parties in the scheme is $n-1$ and the access
structure has the following properties:
\begin{itemize}
\item[$\bullet$] Any group of size less than $d-1$    cannot be used to recover the
secret.
\item[$\bullet$]  There are $\mu_i(p,e1) \mu_j(q,e1)$ pairs of groups of size $(p-1,q-1)$ in   $\Gamma_{H_i}\times \Gamma_{H_j}$, $i \ne j$,  that can  recover
the secret, where $e1$ is any vector of weight 1 in $GF(4)^n$ .

\item[$\bullet$] When the parties come together, up to $\lfloor \frac{d-1}{2} \rfloor$ cheaters  can be
found in each group.

\item[$\bullet$] $\Gamma$ is a minimal access structure if for every element $(x,y)$ in $\Gamma$, no element of $\Gamma_{H_i}$ and $\Gamma_{H_j}$ are subsets of $x$ and $y$, respectively.
\end{itemize}
\end{theorem}

\begin{proof}
The first property is trivial from the definition of $\Gamma_{H_i}$'s. The minimum size of any group in $\Gamma_{H_i}$ is greater or equal to $d-1$. Thus any group of size  less than $d-1$ cannot be used to recover the secret. We can get the second property from the proof of Lemma \ref{lem2}. Since any element in $\Gamma_{H_i}\times \Gamma_{H_j}$, $i \ne j$, can recover the secret $s$, there are  $\mu_i(p,e1) \mu_j(q,e1)$ pairs of groups of size $(p-1,q-1)$  that can  recover
the secret. The third property comes from the error-correcting capability of additive codes over GF(4). The fourth property is from the definition of the minimal access structure.
\qed
\end{proof}

\begin{corollary}\label{cor1}
The pairs of groups from $\Gamma_{H_i}\times \Gamma_{H_j}$, $i \ne j$,  in this SSS based on an additive self-dual code $\C$ have  the following size distribution generating function  :
\begin{equation}
 \sum_p \sum_q \mu_i(p,e1) \mu_j(q,e1)y^{(p-1,q-1)},
\end{equation}
where $p$ and $q$ denote the weights of codewords.

Furthermore, the size distribution generating function for the access structure is as follows:

\begin{equation}\label{co1}
 \sum_{i \ne j}\sum_p \sum_q \mu_i(p,e1) \mu_j(q,e1)y^{(p-1,q-1)}.
\end{equation}

\end{corollary}

\begin{proof}
Since $|\Gamma_{H_i}|=\sum_p \mu_i(p,e1)$,  $|\Gamma_{H_i}\times \Gamma_{H_j}|=\sum_p \sum_q \mu_i(p,e1) \mu_j(q,e1)$  for $i \ne j$.
From this we get the size distribution generating function  of $\Gamma_{H_i}\times \Gamma_{H_j}$.
Since $|\cup_{i \ne j} \Gamma_{H_i}\times \Gamma_{H_j}|=\sum_{i \ne j}\sum_p \sum_q \mu_i(p,e1) \mu_j(q,e1)$, we get the size distribution generating function for the access structure as  (\ref{co1}).\qed

\end{proof}

Note that we redefined minimal access group of SSSs from additive codes over $GF(4)$ considering its distinct way of recovering the secret $s$. Thus we can now develop Theorem \ref{mini} for SSSs based on additive codes over $GF(4)$.
\begin{theorem}
Let $C$ be an $(n,2^k)$ code over $GF(4)$, and let  $G=(\bold{g}_0, \bold{g}_1,\cdots, \bold{g}_{n-1})$ be its generator matrix.  If each nonzero codeword of $C$ is a minimal vector, then in the secret sharing scheme based on $C^{\perp}$, there are altogether $3\cdot 2^{2k-4}$ minimal access groups if the secret $s\in GF(4)$. In addition, we have the following:
\begin{itemize}
\item[$\bullet$] If $\bold{g}_i$ is the same vector to  $\bold{g}_0$, $1\le i \le n-1$, then participant $P_i$ must be in every  $\Gamma_{H_k}$, $1\le k\le 3$. Such a participant is called a $dictatorial$  $participant$ in SSS based on GF(4).

\item[$\bullet$] If $\bold{g}_i$ is not same to  $\bold{g}_0$, $1\le i \le n-1$, then participant $P_i$ must be in $3^3\cdot 2^{2k-8}$ out of $3\cdot 2^{2k-4}$ minimal access groups.

\end{itemize}
\end{theorem}

\begin{proof}
At the beginning of this section, we assumed that none of $\bold{g}_i$'s is the zero vector. Hence $\bold{g}_0\ne 0$. Thus $\bold{ug}_0$ takes on each element of $GF(4)$ exactly $2^{k-2}$ times when $\bold{u}$ ranges over all elements of $GF(2)^k$. Hence there are $2^k-2^{k-2}$ codewords in $C$ with an nonzero component in its first coordinate. Since each nonzero codeword is a minimal vector, a codeword c-covers another one  if and only if they are the same vector. Hence the total number of minimal codewords is $2^k-2^{k-2}=3\cdot 2^{k-2}$.

Since
\begin{equation*}
\begin{split}
|\Go=\{\bold{c}| c_0= \bold{ug}_0&=1,~~\bold{c}\in C\}|=|\Gt=\{\bold{c}| c_0= \bold{ug}_0=\om,~~\bold{c}\in C\}|\\
&=|\Gtt=\{\bold{c}| c_0= \bold{ug}_0=\ob,~~\bold{c}\in C\}|=\frac{3\cdot 2^{k-2}}{3}=2^{k-2},\\
\end{split}
\end{equation*}
the number of minimal access groups is
\begin{equation*}
\sum_{i\ne j}|\Gamma_{H_i}\times \Gamma_{H_j}|= 3\times 2^{k-2}\times 2^{k-2}= 3\cdot 2^{2k-4}.
\end{equation*}

If $\bold{g}_i=\bold{g}_0$, $1\le i \le n-1$, then  $\bold{ug}_0=a\ne 0$ implies $\bold{ug}_i=a$. Thus the participant $P_i$ is involved in every $\Gamma_{H_k}$, $1\le k\le 3$. If $\bold{g}_0$ and $\bold{g}_1$ are linearly independent, then $(\bold{ug}_0, \bold{ug}_i)$ takes on each element of $GF(4)^2$ exactly $3^2\cdot 2^{k-4}$ when the vector $\bold{u}$ ranges over $GF(4)^k$. Hence
\begin{equation*}
|\{\bold{u}| \bold{ug}_0\ne 0,~~\textrm{and}~\bold{ug}_i\ne 0\}|=3^2\cdot 2^{k-4}
\end{equation*}
and
\begin{equation*}
\begin{split}
|\{\bold{u}| \bold{ug}_0=1,~~\textrm{and}~\bold{ug}_i\ne 0\}|&=|\{\bold{u}| \bold{ug}_0=\om,~~\textrm{and}~\bold{ug}_i\ne 0\}|\\
=|\{\bold{u}| \bold{ug}_0=\ob,~~\textrm{and}~\bold{ug}_i\ne 0\}|&=3\cdot 2^{k-4}.\\
\end{split}
\end{equation*}
Thus the number of minimal groups in which $P_i$ is involved is
\begin{equation*}
3\times 3\cdot 2^{k-4} \times 3\cdot 2^{k-4}=3^3\cdot 2^{2k-8}.
\end{equation*}\qed
\end{proof}

The accessibility degree for SSSs based on additive codes can be defined as the follows.
\begin{definition}\label{def5}
\rm{The $accessibility~ index$ on $P$ is the map  $\delta_{P} : \mathcal{A}_P  \longrightarrow \mathbb{R}$    given by}
\begin{equation*}
\delta_{P}(\Gamma)=\frac{|\Gamma|}{2^{2m}}~\textrm{for}~\Gamma \in \mathcal{A}_P
\end{equation*}
\rm{where $m=|P|$, the number of participants. The number $\delta_{P}(\Gamma)$  will be called the $accessibility~ degree$ of structure $\Gamma$.}
\end{definition}
Here we have to divide $|\Gamma|$ by $2^{2m}$ since  this is 2-step SSS and we have to pool the participants' shares twice.

Let $\Gamma$ be the access structure above. Then we can determine the accessibility degree of access structure $\Gamma$  for SSS based on an additive code over $GF(4) $ by
\begin{equation*}
\delta_{P}(\Gamma)=\frac{1}{2^{2n-2}} \sum_{i \ne j}\sum_p \sum_q \mu_i(p,e1) \mu_j(q,e1).
\end{equation*}

\begin{example}
We will describe SSS using the $(6,2^6)$ hexacode.
Let $\mathcal{G}_6$ be a linear $[6,3,4]$ $hexacode$ over $GF(4)$ whose generator matrix is as follows :
\[
{\setlength{\arraycolsep}{2pt}
\renewcommand{\arraystretch}{0.8}
\left[ \begin{array}{cccccc}
1 &0 &0  &1  &\om &\om  \\
 0 & 1  &0  &\om  &1  &\om  \\
 0 & 0  &1  &\om  &\om  &1   \\
\end{array}
\right].}\]
The weight enumerator of  the hexacode $\mathcal{G}_6$ is :
\begin{equation*}
1~+~ 45y^4~+~18y^6.
\end{equation*}
The vectors of weight 4 in $\mathcal{G}_6$ hold a 2-design by Theorem \ref{theorem1} and the vectors of weight 6 hold 1-design  by  Lemma 3.2 in \cite{Dou}. Note that $45=\lambda_2\binom{6}{2}/\binom{4}{2}$, whence $\lambda_2=18$. Thus $\lambda_1=18\binom{5}{1}/\binom{3}{1}=30$. It is easy to see that there are  10 supports of blocks in $\Gamma_{H_1}$, considering scalar multiplication. That is, these numbers can be obtained by dividing $\lambda$ for the 1-design held by these vectors by 3. The following is the size  distribution of the access structure of the hexacode $\mathcal{G}_6$.
\begin{equation}
\sum_{i\in \{4,6\}}\lambda_1(D_i)y^{i-1}=10y^3+6y^5.
\end{equation}

The accessibility degree of the access structure for the linear hexacode  $\mathcal{G}_6$ is
\begin{equation*}
\delta_{P}(\Gamma)=\frac{|\Gamma|}{2^{m}}=\frac{16}{2^5}=\frac{1}{2}=0.5.
\end{equation*}

Now let us think of  $\mathcal{G}_6$ as an additive code. Then it has the following generator matrix:
\[
{\setlength{\arraycolsep}{2pt}
\renewcommand{\arraystretch}{0.8}
\left[ \begin{array}{cccccc}
1 &0 &0  &1  &\om &\om  \\
\om& 0&0&\om&\ob &\ob \\
 0 & 1  &0  &\om  &1  &\om  \\
0&\om&0&\ob &\om&\ob \\
 0 & 0  &1  &\om  &\om  &1   \\
0&0&\om&\ob&\ob&\om \\
\end{array}
\right].}\]

Note that there are two kinds of blocks for extremal additive even self-dual codes : one is a simple block and the other is of  multiplicity 3. By Theorem \ref{Larkt}, it is easy to check  that the supports of weight 4 vectors  form a 2-design with possibly repeated blocks and the supports of weights 4 and 6 vectors form  1-designs by Lemma \ref{lem1}.    Since the hexacode $\mathcal{G}_6$ is  extremal even additive self-dual,   the weight $4$ and $6$ codewords hold  generalized 2-designs of type 3  by Corollary \ref{Hoh}. These codewords eventually hold generalized 1-designs of type 3 by  the definition of generalized $t$-designs. It implies that there are same number of blocks in each $\Gamma_{H_i}$. For example,  when  $\lambda_1(D_4)=10$,  $\mu_i(4,e1)=10$. Similarly, $\mu_i(6,e1)=6$ when $\lambda_1(D_6)=6$.  The supports of the vectors are described in Table 1.

\begin{table}[h]\label{table1}
\centering
\caption{The supports of each weight  for the hexacode $\mathcal{G}_6$ }
\begin{tabular}[c]{|c|c|c|c|}
\hline
 &$\Go(x_0=1)$ &$\Gt(y_0=\om)$ &$\Gtt(z_0=\ob)$\\ \hline
\multirow{10}{*}{wt4} &  $\{ 2,3,4 \}$ & $ \{ 2,3,4 \}$ & $ \{ 2,3,4 \}$\\
& $ \{ 2,3,5 \}$ & $\{ 2,3,5 \}$ & $\{ 2,3,5 \}$\\
& $ \{ 2,3,6 \}$ & $ \{ 2,3,6 \}$ & $ \{ 2,3,6 \}$\\
& $ \{ 2,4,5 \}$ &$ \{ 2,4,5 \}$ &$ \{ 2,4,5 \}$ \\
& $ \{ 2,4,6 \}$ &$ \{ 2,4,6 \}$ &$ \{ 2,4,6 \}$ \\
& $ \{ 2,5,6 \}$ &$ \{ 2,5,6 \}$ &$ \{ 2,5,6 \}$ \\
& $ \{ 3,4,5 \}$ &$ \{ 3,4,5 \}$ &$ \{ 3,4,5 \}$ \\
& $ \{ 3,4,6 \}$ &$ \{ 3,4,6 \}$ &$ \{ 3,4,6 \}$ \\
& $ \{ 3,5,6 \}$ &$ \{ 3,5,6 \}$ &$ \{ 3,5,6 \}$ \\
& $ \{ 4,5,6 \}$ & $ \{ 4,5,6 \}$ & $ \{ 4,5,6 \}$ \\ \hline
$\#$ of wt 4& 10&10&10\\ \hline \hline
wt 6& $\{2,3,4,5,6\}$ &  $\{2,3,4,5,6\}$ &  $\{2,3,4,5,6\}$ \\ \hline
$\#$ of wt 6& 6&6 &6 \\ \hline \hline
Total $\#$ &16&16&16\\ \hline
\end{tabular}
\end{table}

The size distribution of the access structure of the hexacode $\mathcal{G}_6$ by Corollary \ref{cor1} is
\begin{equation}
\begin{split}
\sum_{\substack{i \ne j\\ 1\le i,j\le 3}}\sum_{p\in \{4,6\}} &\sum_{q\in \{4,6\}} \mu_i(p,e1) \mu_j(q,e1)y^{(p-1,q-1)}\\
&=100y^{(3,3)}+60y^{(3,5)}+60y^{(5,3)}+36y^{(5,5)}.
\end{split}
\end{equation}
These pairs of groups comprise the 256 elements of the access structure.
Additionally, a group of size 6  do not c-cover  any group of size 4 . If a weight 4 vector were c-covered by a weight 6 vector, then the sum of the two vectors will yield a weight 2 vector, which is a contradiction. Thus 256 pairs of supports form the minimal access structure.
Note that $\Go$ in Table 1  is the access structure for linear hexacode $\mathcal{G}_6$.

We summarize  the following properties of SSS using  the hexacode.\\
\begin{itemize}
\item[]  $\bold{Summary}$ :
\item[] We summarize  the following properties of SSS using  the hexacode.
\item[(i)]  The access structure consists of 100 pairs of sets of size (3,3), 60 pairs of groups of size (3,5),  60 pairs of groups of size (5,3),
36 pairs of groups of size (5,5).
\item[(ii)]  All the pairs of groups constitute the minimal access structure.
\item[(iii)]  No group of size less than 3 can be used in recovering the secret.
\end{itemize}

Moreover, the accessibility degree for the access structure of the additive hexacode  $\mathcal{G}_6$ is
\begin{equation*}
\delta_{\mathcal{P}}(\Gamma)=\frac{1}{2^{2m}} \sum_{i \ne j}\sum_p \sum_q \mu_i(p,e1) \mu_j(q,e1)=\frac{256}{2^{10}}=\frac{1}{4}=0.25.
\end{equation*}

\end{example}

\begin{example}
Let us consider  a self-dual [12,6,4] code $E_{12}$ with a generator matrix \cite{Mac}
\[
{\setlength{\arraycolsep}{2pt}
\renewcommand{\arraystretch}{0.8}
\left[ \begin{array}{cccccccccccc}
1&1&1&1&0&0&0&0&0&0&0&0\\
0&0&1&1&1&1&0&0&0&0&0&0\\
0&0&0&0&1&1&1&1&0&0&0&0\\
0&0&0&0&0&0&1&1&1&1&0&0\\
0&0&0&0&0&0&0&0&1&1&1&1\\
1&0&1&0&1&0&1&0&1&0&1&0\\
\end{array}
\right].}\]
The weight enumerator of $E_{12}$ is :
\begin{equation}
1+45y^4+216y^6+1755y^8+1800y^{10}+279y^{12}.
\end{equation}
 Since we cannot apply Theorem \ref{theorem1} to $E_{12}$, we have to get the access structure by  MAGMA, which is one of  the commonly used computer languages. Using this, we obtain the following size distribution of access structure for $E_{12}$ :
\begin{equation}
5y^3+36y^5+390y^7+500y^9+93y^{11}.
\end{equation}
The accessibility degree of the access structure  for the SSS based on $E_{12}$  is
 \begin{equation*}
\delta_{\mathcal{P}}(\Gamma)=\frac{|\Gamma|}{2^m}=\frac{1024}{2^{11}}=\frac{1}{2}=0.5.
\end{equation*}

Now we will describe an SSS based on an extremal additive even self-dual $(12,2^{12})$ dodecacode $QC\_12$ (see \cite{Hoh},~\cite{Huff},~\cite{Lark}).
It  has the following generator matrix
\[
{\setlength{\arraycolsep}{2pt}
\renewcommand{\arraystretch}{0.8}
\left[ \begin{array}{cccccccccccc}
0&0&0&0&0&0&1&1&1&1&1&1\\
0&0&0&0&0&0&\om&\om&\om&\om&\om&\om\\
1&1&1&1&1&1&0&0&0&0&0&0\\
\om&\om&\om&\om&\om&\om&0&0&0&0&0&0\\
0&0&0&1&\om&\ob&0&0&0&1&\om&\ob\\
0&0&0&\om&\ob&1&0&0&0&\om&\ob&1\\
1&\ob&\om&0&0&0&1&\ob&\om&0&0&0\\
\om&1&\ob&0&0&0&\om&1&\ob&0&0&0\\
0&0&0&1&\ob&\om&\om&\ob&1&0&0&0\\
0&0&0&\om&1&\ob&1&\om&\ob&0&0&0\\
1&\om&\ob&0&0&0&0&0&0&\ob&\om&1\\
\ob&1&\om&0&0&0&0&0&0&1&\ob&\om\\
\end{array}
\right].}\]
The weight enumerator of  $QC\_12$ is :
\begin{equation*}
1~+~ 396y^6~+~1485y^8~+~1980y^{10}~+~234y^{12}.
\end{equation*}

 By Corollary \ref{Larkc},  the supports of weight 6 vectors forms a 5-design with possibly repeated blocks.  We remark
that there are 18 codewords of weight 6 whose supports  repeat three
times. Among them, exactly 12 codewords are such that their scalar multiples are
also codewords. Since $A_6=396$, there are 396-18=378
codewords of weight 6 whose supports are simple blocks. Note that
$396=\lambda_5\binom{12}{5}/\binom{6}{5}$, whence $\lambda_5=3$. Allowing repeated blocks we obtain a
$5$-$(12,6,\lambda_5=3)$ design in $QC\_12$. Each 5-set is either contained in one
block repeated three times or in three distinct blocks. We see
that vectors of other weights 8, 10, and 12 hold 5-designs with possibly repeated
blocks with $\lambda_5=105,~630$ and $234$, respectively \cite{Lark}.

Since   $QC\_12$ is  extremal even additive self-dual,   the  codewords of all the nonzero weights hold generalized 2-designs of type 3   by Corollary \ref{Hoh} and also hold generalized 1-designs of type 3 by the definition of generalized $t$-designs. Thus we have the following numbers $\mu_i(p,e1)$  by dividing $\lambda_1$  by 3.  When $\lambda_5=3$ for weight 6 codewords, $\lambda_1=3\binom{11}{4}/\binom{5}{4}=198$.   It implies that $\mu_1(6,e1)=\mu_2(6,e1)=\mu_3(6,e1)=66$.  Repeating the calculation for weight 8 with $\lambda_5=105$, we get $\lambda_1=105\binom{11}{4}/\binom{7}{4}=990$. Thus
 $\mu_1(8,e1)=\mu_2(8,e1)=\mu_3(8,e1)=330$. For weight 10 with $\lambda_5=630$, $\lambda_1=630\binom{11}{4}/\binom{9}{4}=1650$. Thus
 $\mu_i(10,e1)=550$. For weight 12 with $\lambda_5=234$, $\lambda_1=234\binom{11}{4}/\binom{11}{4}=234$. Thus
 $\mu_i(12,e1)=78$.

By Corollary \ref{cor1}, the size distribution of the access structure for the dodecacode $QC\_12$  is
\begin{equation}\label{e5}
\begin{split}
&\sum_{\substack{i \ne j\\ 1\le i,j\le 3}}\sum_{p\in \{6,8,10,12\}} \sum_{q\in \{6,8,10,12\}} \mu_i(p,e1) \mu_j(q,e1)y^{(p-1,q-1)}\\
&=4356y^{(5,5)}+21780y^{(5,7)}+36300y^{(5,9)}+5148y^{(5,11)}+21780y^{(7,5)}\\
&+108900y^{(7,7)}+181500y^{(7,9)}+25740y^{(7,11)}+36300y^{(9,5)}
+181500y^{(9,7)}\\
&+302500y^{(9,9)}+42900y^{(9,11)}+5148y^{(11,5)}+25740y^{(11,7)}+42900y^{(11,9)}\\
&+6084y^{(11,11)}.\\
\end{split}
\end{equation}

A vector of weight 8  does not c-cover  any vector of weight 6. If a weight 6 vector were c-covered by a weight 8 vector, then the sum of the two vectors will yield a weight 2 vector, which is a contradiction. Likewise, a vector of weight 10  does not c-cover  any vector of weight 6 or 8 . If a weight 6 or 8 vector were c-covered by  a weight 10 vector, then the sum of the two vectors will yield a weight 4, or 2 vector, respectively, which is a contradiction.

\begin{itemize}
\item[]  $\bold{Summary}$ :
\item[] We summarize  the following properties of SSS using  $QC\_12$.
\item[(i)] The access structure consists of the pairs of groups as in (\ref{e5}).
\item[(ii)] All the pairs of groups with the sizes $\in \{5,7,9\}$ are contained in the minimal access structure.
\item[(iii)] No group of size less than 5 can be used in recovering the secret.
\end{itemize}

The accessibility degree for the SSS based on the dodecacode $QC\_12$ is
\begin{equation*}
\delta_{\mathcal{P}}(\Gamma)=\frac{1}{2^{2m}} \sum_{i\ne j}\sum_p \sum_q \mu_i(p,e1) \mu_j(q,e1)=\frac{1,048,576}{2^{22}}=\frac{1}{4}=0.25.
\end{equation*}

\end{example}

\begin{example}
Now we are going to describe an SSS based on $S_{18}$ which is  an $(18, 2^{18})$ extremal additive even self-dual code. The weight enumerator of $S_{18}$ is \cite{Mac}:
\begin{equation}
1+2754y^8+18360y^{10}+77112y^{12}+110160y^{14}+50949y^{16}+2808y^{18}.
\end{equation}
Since all the non-zero weights in  $S_{18}$ hold $5$-designs with possibly repeated blocks by Corollary \ref{Larkc}, we can easily calculate $\lambda_1$ for each weight using (\ref{eq2}).

Note that
$2754=\lambda_5\binom{18}{5}/\binom{8}{5}$, whence $\lambda_5=18$. Allowing repeated blocks we obtain a
$5$-$(18,8,\lambda_5=18)$ design in $S_{18}$. Each $5$-set is either contained in one
block repeated three times or in three distinct blocks. We see
that vectors of other weights 10, 12, 14, 16, and 18 hold 5-designs with possibly repeated
blocks with $\lambda_5=540,~7128,~25740,~ 25974,$ and $2808$, respectively.

Since   $S_{18}$ is  extremal even additive self-dual,   the  codewords of all the nonzero weights hold generalized 2-designs of type 3  by Corollary \ref{Hoh}  and also hold generalized 1-designs of type 3 by the definition of generalized $t$-designs. Thus we have the following numbers $\mu_i(p,e1)$  by dividing $\lambda_1$  by 3.  When $\lambda_5=18$ for weight 8 codewords, $\lambda_1=18\binom{17}{4}/\binom{7}{4}=1224$.   It implies that $\mu_1(8,e1)=\mu_2(8,e1)=\mu_3(8,e1)=408$.  Repeating the calculation for weight 10 with $\lambda_5=540$, we get $\lambda_1=540\binom{17}{4}/\binom{9}{4}=10200$. Thus
 $\mu_1(10,e1)=\mu_2(10,e1)=\mu_3(10,e1)=3400$.
For weight 12 with $\lambda_5=7128$, $\lambda_1=7128\binom{17}{4}/\binom{11}{4}=51408$. Thus
 $\mu_i(12,e1)=17136$. For weight 14 with $\lambda_5=25740$, $\lambda_1=25740\binom{17}{4}/\binom{13}{4}=85680$. Thus
 $\mu_i(14,e1)=28560$.  For weight 16 with $\lambda_5=25974$, \\
$\lambda_1=25974\binom{17}{4}/\binom{15}{4}=45288$. Thus
 $\mu_i(16,e1)=15096$. For weight 18 with $\lambda_5=2808$, $\lambda_1=2808\binom{17}{4}/\binom{17}{4}=2808$. Thus
 $\mu_i(18,e1)=936$.

By Corollary \ref{cor1}, the size distribution of the access structure for  $S_{18}$  is
{\small
\begin{equation}\label{e6}
\begin{split}
&\sum_{\substack{i \ne j\\ 1\le i,j\le 3}}\sum_{p\in \{8,10,12, 14, 16, 18\}} \sum_{q\in \{8,10,12, 14, 16, 18\}} \mu_i(p,e1) \mu_j(q,e1)y^{(p-1,q-1)}\\
&=166464y^{(7,7)}+1387200y^{(7,9)}+6991488y^{(7,11)}+11652480y^{(7,13)}+6159168y^{(7,15)}\\
&+381888y^{(7,17)}+1387200y^{(9,7)}+11560000y^{(9,9)}+58262400y^{(9,11)}
+97104000y^{(9,13)}\\
&+51326400y^{(9,15)}+3182400y^{(9,17)}+6991488y^{(11,7)}+58262400y^{(11,9)}\\
&+293642496y^{(11,11)}+489404160y^{(11,13)}+258685056y^{(11,15)}+16039296y^{(11,17)}\\
&+11652480y^{(13,7)}+97104000y^{(13,9)}+489404160y^{(13,11)}+815673600y^{(13,13)}
\\
&+431141760y^{(13,15)}+26732160y^{(13,17)}+6159168y^{(15,7)}+51326400y^{(15,9)}
\\
&+258685056y^{(15,11)}+431141760y^{(15,13)}+227889216y^{(15,15)}+14129856y^{(15,17)}\\
&+381888y^{(17,7)}+3182400y^{(17,9)}+16039296y^{(17,11)}+26732160y^{(17,13)}\\
&+14129856y^{(17,15)}+876096y^{(17,17)}.\\
\end{split}
\end{equation}
}

A vector of weight 10  does not c-cover  any vector of weight 8. If a weight 8 vector were c-covered by a weight 10 vector, then the sum of the two vectors will yield a weight 2 vector, which is a contradiction. Likewise, a vector of weight 12  does not c-cover  any vector of weight 8 or 10 . If a weight 8 or 10 vector were c-covered by  a weight 12 vector, then the sum of the two vectors will yield a weight 4, or 2 vector, respectively, which is a contradiction.
 A vector of weight 14  does not c-cover  any vector of weight 8 ,10 or 12 . If a weight 8, 10  or 12 vector were c-covered by  a weight 14 vector, then the sum of the two vectors will yield a weight 6, 4, or 2 vector, respectively, which is a contradiction.

\begin{itemize}
\item[]  $\bold{Summary}$ :
\item[] We summarize  the following properties of SSS using  $S_{18}$.
\item[(i)] The access structure consists of the pairs of groups as in (\ref{e6}).
\item[(ii)] All the pairs of groups with the sizes $\in \{7,9,11,13\}$ are contained in the minimal access structure.
\item[(iii)] No group of size less than 7 can be used in recovering the secret.
\end{itemize}

The accessibility degree for the SSS based on $S_{18}$  is
\begin{equation*}
\delta_{\mathcal{P}}(\Gamma)=\frac{1}{2^{2m}} \sum_{i\ne j}\sum_p \sum_q \mu(p,e1) \mu(q,e1)=\frac{4294967296}{2^{34}}=\frac{1}{4}=0.25.
\end{equation*}
The accessibility degree of the access structure  for the SSS based on $S_{18}$  is $\frac{1}{4}$ which is same as that of  the dodecacode $QC\_12$.

\end{example}

\section{Conclusion}

In this paper,  we introduce two contrasting access structures, one from linear codes and the other from additive codes. The new results we obtained are mainly stated in Section 3.2 and 3.3. In Section 3.2, we newly defined SSSs based on additive codes over $GF(4)$.  The access structure from additive codes over $GF(4)$ is described  in a distinct way requiring at least two steps of calculations.   In Section 3.3, we determined the access structure for SSSs based on a hexacode, a dodecacode over $GF(4)$ and $S_{18}$ using the notion we introduced in Section 3.2.

\end{document}